\documentclass[peerreview,onecolumn]{IEEEtran}
\ifCLASSINFOpdf
\else
\fi
\usepackage[english]{babel}
\usepackage[latin1]{inputenc}
\usepackage[T1]{fontenc}

\usepackage{amsmath}
\usepackage{amsfonts}
\usepackage{amssymb}

\usepackage[all]{xy} 

\usepackage{rotating}
\usepackage{multirow}
\usepackage{textcomp}

\usepackage{enumerate}
\usepackage[shortlabels]{enumitem}

\usepackage{algorithm} 
\usepackage[noend]{algorithmic}

\usepackage{graphicx}
\usepackage{makeidx,shortvrb,latexsym}
\usepackage{fancyhdr}
\usepackage{indentfirst}
\usepackage[usenames,dvipsnames]{pstricks}
\usepackage{epsfig}

\usepackage{hyperref} 
\hypersetup{
    colorlinks,
    citecolor=black,
    filecolor=black,
    linkcolor=black,
    urlcolor=black
}

\usepackage{verbatim}
\usepackage{listings} 
\usepackage{colortbl}
\usepackage{xcolor}
\usepackage{color}
\definecolor{mygreen}{rgb}{0,0.6,0}
\definecolor{mygray}{rgb}{0.5,0.5,0.5}
\definecolor{mymauve}{rgb}{0.58,0,0.82}
\newtheorem{theorem}{Theorem}

\newtheorem{proposition}{Proposition}

\newtheorem{definition}{Definition}
\newtheorem{example}{Example}
\newtheorem{remark}{Remark}
\newtheorem{fact}{Fact}
\newtheorem{proof}{Proof}

\newcommand{\B}{{\mathcal B}}

\newcommand{\Fc}{\mathcal{F}_C}

\newcommand{\FF}{\mathbb{F}}
\newcommand{\KK}{\mathbb{K}}

\newcommand{\QQ}{\mathbb{Q}}

\newcommand{\ZZ}{\mathbb{Z}}
\newcommand{\dd}{\mathrm{d}}

\newcommand{\NNF}{\mathrm{NNF}}

\newcommand{\w}{\mathrm{w}}
\newcommand{\wP}{{\mathfrak w}}

\newcommand{\dP}{{\mathfrak d}}

\newcommand{\op}{{\sf p}} 
\newcommand{\Bf}{B.f.$\;$}

\newcommand{\fanf}{f^{(\FF)}}
\newcommand{\fnnf}{f^{(\ZZ)}}

\begin{document}
%
\title{A deterministic algorithm for the distance and weight distribution of binary nonlinear codes}
%
%
%

\author{Emanuele~Bellini,
        and~Massimiliano~Sala. 
\thanks{E. Bellini, Telsy S.p.A., Italy.}
\thanks{e-mail: eemanuele.bellini@gmail.com.}
\thanks{M. Sala, University of Trento, Italy.}
\thanks{email: maxsalacodes@gmail.com}
}
%
%

\markboth{Journal of \LaTeX\ Class Files,~Vol.~??, No.~?, Month~YEAR}%
{Shell \MakeLowercase{\textit{et al.}}: Bare Demo of IEEEtran.cls for Journals}
%



\maketitle

\begin{abstract}
Given a binary nonlinear code, we provide a deterministic algorithm to compute its weight and distance distribution, 
and in particular its minimum weight and its minimum distance,
which takes advantage of fast Fourier techniques.
This algorithm's performance is similar to that of best-known algorithms for the average case, 
while it is especially efficient for codes with low information rate. 
We provide complexity estimates for several cases of interest.
\end{abstract}

%
%
\IEEEpeerreviewmaketitle
%
%
\section{Introduction}
 \label{sec:intro}
Let $C$ be a nonlinear code, that is, a code which is not necessarily linear.
There are some related computational problems which are of interest, that we list as
the computation of: the distance distribution (A), the minimum distance (A1), 
a minimum-distance codeword-pair (A2), the weight distribution (B), the minimum weight (B1),
a minimum-weight codeword (B2).
The decoding performance of $C$ can be established by solving Problem A and can be estimated by 
solving Problem A1. 
\begin{remark} \label{rem}
Solving Problem A2 (respectively, B2) implies solving Problem A1 (B1), but the converse does not hold. However, it is noteworthy that no known algorithm is able to solve
A1 (B1) without solving A2 (B2). 
\end{remark}
If $C$ is linear, Problem A (respectively, A1, A2) and B (B1, B2) are equivalent. This holds also for some nonlinear codes, called distance-invariant codes 
\cite{CGC-cd-art-mitchell1989distance}, 
and many of these are optimal codes (e.g. the Preparata-Kerdock codes \cite{CGC-cd-art-preparata}).
When $C$ is linear, we consider also the decoding problem, which is implied
by solving Problem B2 in the suitable code coset (which is a nonlinear code). 
Observe that the considerations in Remark \ref{rem} remain valid also if we restrict to linear codes.\\
\indent
In the linear case it is convenient to use probabilistic algorithms for the computation of the minimum distance,
such as the Brouwer-Zimmerman algorithm \cite{CGC-cd-book-zimmermann1996}, or any of its variations, e.g. \cite{CGC-cd-art-cant98}.\\
We note that these algorithms must actually retrieve (at least) one minimum-weight codeword in order to obtain the minimum-weight value.\\
%
%
In the nonlinear case the minimum weight and the minimum distance may be different.
For some classes of nonlinear codes there are algorithms which perform much better than brute force, e.g. code with large kernel (\cite{CGC-cd-art-pujol2012minimum,CGC-cd-art-villanueva2014efficient}) or additive codes (\cite{CGC-cd-art-white2006new}). 
However, in the general nonlinear case it is not possible to improve significantly on the brute-force approach, as shown in \cite{CGC-cd-art-elemanumax}. 
Indeed, we are not aware of any non-exponential probabilistic or deterministic algorithm to solve any of the problems A, A1, A2, B, B1, B2.
In particular, to compute the weight distribution of a generic binary $(n,2^k)$-nonlinear code \emph{given} as a list of binary vectors,
we need to perform $O(n2^k)$ bit operations,
while finding the distance distribution requires $O(n2^{2k})$ bit operations. \\
The main result of this paper is a \emph{deterministic} algorithm to compute the distance and weight distribution, and thus the minimum distance and the minimum weight, 
of any random binary code \emph{represented} as a set of Boolean functions in numerical normal form (NNF). 
Our method performs better than brute force for those codes with low information rate and sparse NNF representation, 
while in the general case, it achieves the same asymptotic computational complexity as brute force methods.\\
In Section \ref{sec:prel}, after some preliminaries on Boolean functions, 
we argue that representing a code as a set of Boolean functions in NNF
does not have any particular drawback with respect to the classical representation of a code as a set of binary vectors. 
In Section \ref{sec:WordsOfWeight_t}, to each binary code we associate a polynomial whose evaluations are the weights of the code.
Similarly, in Section \ref{sec:PairsOfDistance_t}, to each binary code we associate a polynomial whose evaluations are the distances of all possible pairs of codewords.
Given these two polynomials we are able to compute the weight and the distance distribution of any binary nonlinear code.
Finally, in Section \ref{sec:WeightComplexity} we provide some complexity considerations regarding our algorithms.
In particular, we show that, 
to compute the weight distribution starting from the NNF representation of a binary nonlinear code has a complexity of $O((n/h+k)2^k)$, where $n/h$ is the average number of nonzero monomials of the Boolean functions representing the code. 
Moreover, there are many important cases where our approach is provably faster than brute-force (e.g. in the linear case and in the nonlinear case when the NNF representation of the code is sparse), and cases where it is experimentally faster than the Brouwer-Zimmerman method.
\section{Preliminaries}
 \label{sec:prel}
\subsection{Representations of Boolean functions}
%
In this section we briefly summarize some definitions and known results from \cite{CGC-cd-book-carlet} and \cite{CGC-cd-book-macwilliamsI}, concerning representations of Boolean functions.\\
We denote by $\FF$ the field $\FF_2$. The set $\FF^n$ is the set of all binary vectors of length $n$, viewed as an $\FF$-vector space.\\
A \emph{Boolean function} (\Bf) is a function $f:\FF^n\rightarrow \FF$. The set of all Boolean functions from $\FF^n$ to $\FF$ will be denoted by $\B_n$.
%
There are several ways one can uniquely represent a \Bf. We briefly outline those we need.
\subsubsection{Evaluation vector}
We assume to have ordered $\FF^n$, so that $\FF^n=\{\op_1,\ldots,\op_{2^n}\}$. 
A Boolean function $f$ can be specified by a \emph{truth table}, which gives the evaluation of $f$ at all $\op_i$'s.
We consider the evaluation map:
$$
\B_n \longrightarrow \FF^{2^n} 
\qquad
f \longmapsto \underline{f}=(f(\op_1),\ldots,f(\op_{2^n}))\,.
$$
The vector $\underline{f}$ is called the \emph{evaluation vector} of $f$.
Once the order on $\FF^n$ is chosen, i.e. the $\op_i$'s are fixed, it is clear that the evaluation vector of $f$ identifies $f$. 
\subsubsection{Algebraic normal form}\label{secANF}
A Boolean function $f\in\B_n$ can be expressed in a unique way as a square-free polynomial in $\FF[X]=\FF[x_1,\ldots,x_n]$, i.e.
$$f=\sum_{v \in \FF^n}b_vX^v\,,$$
where $X^v=x^{v_1}\cdots x^{v_n}$.\\
This representation is called the \emph{Algebraic Normal Form} (ANF).\\
%
There exists a simple divide-and-conquer butterfly algorithm (\cite{CGC-cd-book-carlet}, p.10) to compute the ANF from the truth-table (or vice-versa) of a Boolean function, which requires $O(n2^{n})$ bit sums (with big $O$ constant $1/2$), while $O(2^n)$ bits must be stored. This algorithm is known as the \emph{fast M\"obius transform}.
\subsubsection{Numerical normal form}\label{secNNF}
%
In \cite{CGC-cry-art-carlet1999} (see also \cite{CGC-cry-art-carlet2001bent}, \cite{CGC-cry-carlet2002coset}) the following representation of Boolean functions has been introduced.\\
Let $f$ be a function on $\FF^n$ taking values in a field $\KK$. We call the \emph{numerical normal form (NNF)} of $f$ the following expression of $f$ as a polynomial:
$$
f(x_1,\ldots,x_n) = \sum_{u \in \FF^n}\lambda_u (\prod_{i=1}^{n}x_i^{u_i}) = \sum_{u \in \FF^n}\lambda_{u}X^u\,,
$$
with $\lambda_{u} \in \KK$ and $u=(u_1,\ldots,u_n)$.\\
It can be proved 
(\cite{CGC-cry-art-carlet1999}, Proposition 1)
that any Boolean function $f$ admits a unique numerical normal form.
As for the ANF, it is possible to compute the NNF of a Boolean function from its truth table by mean of an algorithm similar to a fast Fourier transform, thus requiring $O(n2^n)$ additions over $\KK$ and storing $O(2^n)$ elements of $\KK$.\\
\indent
From now on let $\KK = \QQ$.\\
The truth table of $f$ can be recovered from its NNF by the formula $$f(u)=\sum_{a\preceq u}\lambda_a,\forall u \in \FF^n\,,$$
where $a\preceq u\iff \forall i \in \{1,\ldots,n\} \; a_i \le u_i$. Conversely, 
as shown in \cite{CGC-cry-art-carlet1999} (Section 3.1),
it is possible to derive an explicit formula for the coefficients of the NNF by means of the truth table of $f$.
\begin{proposition}\label{propNNFcoeffIntro}
 Let $f$ be any integer-valued function on $\FF^n$. For every $u\in \FF^n$, the coefficient $\lambda_u$ of the monomial $X^u$ in the NNF of $f$ is:
 \begin{equation}\label{eqNNFCoeff}
  \lambda_u = (-1)^{\w(u)}\sum_{a\in \FF^n |a\preceq u}(-1)^{\w(a)}f(a)\,.
 \end{equation}
\end{proposition}
It is possible to convert a Boolean function from NNF to ANF simply by reducing its coefficients modulo 2. 
The inverse process is less trivial. One can either apply Proposition \ref{propNNFcoeffIntro} to the evaluation vector of $f$ or apply recursively the fact that 
\begin{align}\label{eq:ANFtoNNF}
 a +_{\FF} b = a +_{\ZZ} b +_{\ZZ} (-2ab)\,,
\end{align}
and the fact that each variable has to be square-free (we are working in the affine algebra $\KK[x_1,\cdots,x_n]/\langle x_1^2-x_1,\cdots,x_n^2-x_n \rangle$).
%
%
\subsection{Representing a code as a set of Boolean functions}
  \label{sec:CodeRepresentation}
We consider binary codes, i.e. codes over the finite field $\FF$ of length $n$, with $M$ codewords. 
A binary code $C$ with such parameters is denoted as a $(n,M)$-code. If the code is a subspace of dimension $k$ of $(\FF)^n$ then it is called linear and we indicate it as a $[n,k]$-linear code.\\
Now we show that any binary $(n,2^k)$-code $C$ with $2^k$ codewords can be represented in a unique way as a set of $n$ Boolean functions $f_1, \ldots, f_n : (\FF)^k \to \FF$. 
We indicate with $\fanf$ a Boolean function represented in algebraic normal form, and with $\fnnf$ a Boolean function represented in numerical normal form.
\begin{definition}
 Given a binary $(n,2^k)$-code $C$, consider a fixed order of the codewords of $C$ and of the vectors of $(\FF)^k$. Then consider the matrix $M$ whose rows are the codewords of $C$. We call the \emph{defining polynomials} of the code $C$ the set $\Fc = \{f_1,\ldots,f_n\}$ of the uniquely determined Boolean functions whose truth table are the columns of $M$. We also indicate with $F=(f_1,\ldots,f_n) \in \FF[X]^n$, where $X=x_1,\ldots,x_k$, the polynomial vector whose components are the defining polynomials of $C$.
 With abuse of notation, we sometimes write
 \begin{align*}
  \Fc = \{\fanf_1,\ldots,\fanf_n\} \text{ or } \Fc = \{\fnnf_1,\ldots,\fnnf_n\}
 \end{align*}
\end{definition}
Notice that $F$ can be seen as an encoding function, since $F:(\FF)^k \to (\FF)^n$.
%
%
\subsubsection{Memory cost of representing a code}
 Let us call \emph{vectorial} the representation of a code as a list of vectors over $\FF$, and \emph{Boolean} the representation of the same code as a list of Boolean functions.\\
 For a random code, in terms of memory cost, the two representations are equivalent. 
 In the vectorial representation we need to store all the components of each codeword, which are $n$ times $2^k$ codewords. In the Boolean representation we need to store the $2^k$ coefficients of the $n$ defining polynomials. In both cases we need a memory space of order $O(n2^k)$.\\
 If the code $C$ is linear it can be represented with a binary generator matrix of size $k \times n$. In this case the defining polynomials are linear Boolean functions, i.e. any is of the form
 $
 \sum_{i=1}^{k} \lambda_i x_i, \lambda_i \in \FF
 \,,$ 
 which means that to represent them it is sufficient to store $kn$ elements of $\FF$, yielding again an equivalent representation.\\
 As shown in \cite{CGC-cd-art-pujol2012minimum,CGC-cd-art-villanueva2014efficient}, if $C$ is a binary code of length $n$ with kernel $K$ of dimension $k_K$ and $t$ coset leaders given by the set $S=\{c_1,\ldots,c_t\}$, we can represent it as the kernel $K$ plus the coset leaders $S$. Since the kernel needs a memory space of order $O(nk_K)$, then the kernel plus the $t$ coset leaders takes up a memory space of order $O(n(k_K+t))$.
 When $C$ is linear then $C=\ker(C)$, so the generator matrix is used to represent $C$. On the other hand, when $t+1 = |C|$, then representing the code as the kernel plus the coset leaders requires a memory of $O(n|C|) = O(n2^k)$ 
 (since we are supposing the code has $2^k$ codewords). 
 In the latter case, a Boolean representation could be more convenient.
%
 Another situation in which a Boolean representation is more convenient is the case where the dimension $k$ of the code is much less than the length $n$, i.e. when certain components have to be repeated. 
\\
\indent
It is worth noticing that a linear structure of a nonlinear binary code can be found over a different ring.  For example there are binary codes which have a $\ZZ_4$-linear or $\ZZ_2\ZZ_4$-linear structure and, therefore, they can also be compactly represented using quaternary generator matrix, as shown in \cite{CGC-cd-art-hammons} and \cite{CGC-cd-art-Borg2010}.\\
It can be shown that representing a code with ``practical'' parameters and using NNF \Bf's is as convenient as the usual representation of the code. 
%
%
\subsubsection{Number of coefficients of the NNF}
  \label{sec:NumberOfCoeffNNF}
In order to prove that representing a code with practical parameters and using NNF \Bf's is as convenient as the usual representation of the code, in this section we want to study the distribution of the number of nonzero coefficients of a \Bf represented in NNF, i.e., once the number of variables $k$ is fixed we want to know how many \Bf's have only one nonzero coefficient, how many have two, and so on. \\
We are also interested in finding a relation between this distribution and the distribution of the number of nonzero coefficients of a \Bf represented in ANF.\\
In Table \ref{tabNonzeroCoeffDistr} we report the distribution of the nonzero coefficients of \Bf's represented in ANF and NNF with $k=1,2,3,4$ variables. As one may expect, the ANF follows a binomial distribution. This means that choosing a random \Bf its ANF is likely to have half of the coefficients equal to $0$ and half equal to $1$. This does not happen for the NNF, although for $k$ small the two distributions are close. This means that, when $k$ is small, a random binary $(n,2^k)$-nonlinear code can be represented with a set of \Bf's in NNF with half of the coefficients equal to $0$ with high probability, while sparse NNF representations are more rare as $k$ grows.
\begin{table*}[ht]
\begin{center}
\scalebox{0.7}{
\begin{tabular}{l|l llll llll llll llll }
k & 0 & 1 & 2 & 3 & 4 & 5 & 6 & 7 & 8 & 9 & 10 & 11 & 12 & 13 & 14 & 15 & 16 \\ 
\hline
A: 1 & 1 & {\bf 2} & 1 & - & - & - & - & - & - & - & - & - & - & - & - & - & - \\
N: 1 & 1 & {\bf 2} & 1 & - & - & - & - & - & - & - & - & - & - & - & - & - & - \\
A: 2 & 1 & 4 & {\bf 6} & 4 & 1 & - & - & - & - & - & - & - & - & - & - & - & - \\
N: 2 & 1 & 4 & {\bf 5} & 4 & 2 & - & - & - & - & - & - & - & - & - & - & - & - \\
A: 3 & 1 & 8 & 28 & 56 & {\bf 70} & 56 & 28 & 8 & 1 & - & - & - & - & - & - & - & - \\
N: 3 & 1 & 8 & 19 & 42 & {\bf 59} & 50 & 34 & 28 & 15 & - & - & - & - & - & - & - & - \\
A: 4 & 1 & 16 & 120 & 560 & 1820 & 4368 & 8008 & 11440 & {\bf 12870} & 11440 & 8008 & 4368 & 1820 & 560 & 120 & 16 & 1 \\
N: 4 & 1 & 16 & 65 & 304 & 840 & 1768 & 3250 & 5458 & 8077 & {\bf 9986} & 9819 & 7948 & 5954 & 4458 & 3193 & 2830 & 1569 
\end{tabular}
}
\end{center}
\caption{Distribution of the nonzero coefficients in the ANF and NNF.}
\label{tabNonzeroCoeffDistr}
\end{table*}
\begin{proposition}\label{thmNNFNumOfCoeff}
 Let $f$ be a \Bf in $k$ variables. Let $\fanf$ and $\fnnf$ be respectively the ANF and the NNF of $f$. Then if $\fanf$ is a polynomials with $r \le 2^k$ nonzero coefficients, then $\fnnf$ is a polynomial with no more than $\min\{2^k,2^r-1\}$ nonzero coefficients.
\end{proposition}
\begin{proof}
 When computing the NNF from the ANF we have again the $r$ initial terms of the ANF, plus $\binom{r}{2}$ terms which are all possible double product of the $r$ initial terms, plus, in general, $\binom{r}{i}$ terms which are all possible $i$-product of the $r$ initial terms, for each $i \in \{1,\ldots,r\}$. Thus we will have 
 \begin{align}
  \sum_{i=1}^{r} \binom{r}{i} = 2^r-1
 \end{align}
 terms to be summed together. If no sum of similar monomials becomes zero than we have $2^r-1$ nonzero terms.
\end{proof}
By Proposition \ref{thmNNFNumOfCoeff}, if we want a NNF with no more than $s$ terms then we have to choose the ANF with no more than $r = \log_2(s+1)$ terms.
\begin{proposition}\label{thmNNFofLinBF}
 Let $f$ be a linear \Bf in $k$ variables. Let $\fanf$ and $\fnnf$ be respectively the ANF and the NNF of $f$. 
 Thus, for $i_1<i_2<\ldots<ir, r \le k$,
  \begin{align*}
   \fanf = x_{i_1} + \ldots + x_{i_r}\,,
  \end{align*}
 for $r \le k$.
 Then $\fnnf$ is a polynomial with exactly $2^r-1$ nonzero coefficients:
 \begin{align*}
  \fnnf = \sum_{\substack{v \in (\FF)^r \\ v=(v_1,\ldots,v_h) \ne 0}} (-1)^{\w(v)-1} \binom{r}{\w(v)-1} x_{i_1}^{v_1} \cdots x_{i_r}^{v_r}\,.
 \end{align*}
\end{proposition}
\begin{proof}
 Directly from Proposition \ref{propNNFcoeffIntro}.
\end{proof}
Proposition \ref{thmNNFofLinBF} shows that for a linear \Bf, its NNF representation is much denser than its ANF representation.
%
%
\section{Finding the codewords with weight exactly \texorpdfstring{$t$}{Lg}}
  \label{sec:WordsOfWeight_t}
It is possible to construct a polynomial with integer coefficients whose evaluations in $\{0,1\}^k \subseteq \ZZ^k$ are the weights of the codewords of the code $C$.
\begin{definition}\label{def:WeightPolynomial}
 Let $X = \{x_1,\ldots,x_k\}$, and $X^2-X = \{x_1^2-x_1,\ldots,x_k^2-x_k\}$.
 We call the \emph{{\bf weight polynomial}} of the code $C$ the polynomial
 $$\wP_C(X) = \sum_{i=1}^{n}\fnnf_i(X) \in \ZZ[X] / \langle X^2-X \rangle\,,$$
 where the $\fnnf_i$'s are the defining polynomials of the code $C$ in NNF.
\end{definition}
\begin{theorem}
 Let $v \in \{0,1\}^k \subseteq \ZZ^k$. Then there exists a codeword $c \in C$ such that $\w(c) = \wP_C(v)$.
\end{theorem}
\begin{proof}
 It is sufficient to note that $\forall c \in C,c = (\fnnf_1(P),\ldots,\fnnf_n(P))$
 for some $P \in \{0,1\}^k$, and that the sum of all $\fnnf_i$ is over the integers,
 with $\fnnf_i(P) \ge 0$, for $i=1,\ldots,n$.
\end{proof}
Once we have the weight polynomial $\wP_C$ of the code $C$, not only we can find the minimum weight of $C$, but we also find which are the codewords having certain weights by looking at its evaluation vector over the set $\{0,1\}^k$. As we will see in Section \ref{secWPtoEv}, computing this evaluation has a cost of $O(k2^k)$. The complexity maintains the same order if the number of terms of each defining polynomial in NNF is on average $O(\frac{k}{n}2^k)$.\\
\indent
We summarize in Algorithm \ref{algWeightPolyEval} the steps to obtain the weight distribution of a binary $(n,2^k)$-code $C$ given as a list of $2^k$ codewords 
(and thus also the minimum weight of $C$), by finding the evaluation vector of the weight polynomial $\wP_C$. 
We indicate with $C_{i,j}$ the $j$-th component of the $i$-th word of $C$, with $1 \le j \le n$ and $1 \le i \le 2^k$.
\begin{algorithm}[H]
\caption{To find the weight distribution $\underline{\wP}_C$ of a binary nonlinear code $C$.}
\label{algWeightPolyEval}
\begin{algorithmic}[1]
\REQUIRE{$c_1,\ldots,c_{2^k} \in C$}
\ENSURE{the evaluation vector $\underline{\wP}_C$ of $\wP_C$}
\STATE{$\fnnf_{j} \leftarrow \NNF$ of the binary vector $(C_{1,j},\ldots,C_{2^k,j})$ for $1 \le j \le n$}\label{algWPE:stepNNF}
\STATE{$\wP_C \leftarrow \fnnf_{1} + \ldots + \fnnf_{n}$}\label{algWPE:stepWP}
\STATE{$\underline{\wP}_C \leftarrow$ Evaluation of $\wP_C$ over $\{0,1\}^k$}\label{algWPE:stepEV}
\RETURN{$\underline{\wP}_C$} 
\end{algorithmic}
\end{algorithm}
%
%
\section{Finding pairs of codewords with distance exactly \texorpdfstring{$t$}{Lg}}
  \label{sec:PairsOfDistance_t}
It is straightforward to adapt the techniques in Section \ref{sec:WordsOfWeight_t} to the computation of the distance distribution of a code $C$.\\
First, we show how to construct a polynomial with integer coefficients whose evaluations in $\{0,1\}^{2k} \subseteq \ZZ^{2k}$ are the distances of all possible pairs of codewords of the code $C$.
\begin{definition}\label{def:DistancePolynomial}
 Let $X = x_1,\ldots,x_k$, $\tilde{X} = \tilde{x_1},\ldots,\tilde{x_k}$, 
 and $X^2-X = x_1^2-x_1,\ldots,x_k^2-x_k$, $\tilde{X}^2-\tilde{X} = \tilde{x_1}^2-\tilde{x_1},\ldots,\tilde{x_k}^2-\tilde{x_k}$.\\
 We call the \emph{{\bf distance polynomial}} of the code $C$ the polynomial
 \begin{align*}
  \dP_C(X) &= \sum_{i=1}^{n} (\fnnf_i(X)-\fnnf_i(\tilde{X}))^2  \\
           &\in \ZZ[X,\tilde{X}] / \langle X^2-X,\tilde{X}^2-\tilde{X} \rangle\,,
 \end{align*}
 where the $\fnnf_i$'s are the defining polynomials of the code $C$ in NNF.
\end{definition}
Notice that the squaring operation does not introduce squared variables in the expression of $\dP_C$, 
because we are working in the quotient ring $\ZZ[X,\tilde{X}] / \langle X^2-X,\tilde{X}^2-\tilde{X} \rangle$.\\
Notice also that, for $v = (v_1,\ldots,v_k,v_{k+1},\ldots,v_{2k}) \in \{0,1\}^{2k}$, we have that 
$\dP_C((v_1,\ldots,v_k,v_{k+1},\ldots,v_{2k})) = 0$ if and only if $v_i = v_{k+1}$ for $i=1,\ldots,k$, 
and that $\dP_C((v_1,\ldots,v_k,v_{k+1},\ldots,v_{2k})) = \dP_C((v_{k+1},\ldots,v_{2k},v_{1},\ldots,v_{k}))$.
\begin{theorem}
 Let $v \in \{0,1\}^{2k} \subseteq \ZZ^{2k}$ such that $(v_1,\ldots,v_k) \ne (v_{k+1},\ldots,v_{2k})$. 
 Then there exists a pair of distinct codewords $c_1,c_2 \in C$ such that $\dd(c_1,c_2) = \dP_C(v)$.
\end{theorem}
\begin{proof}
 Note that $\forall c_1,c_2 \in C,c_1 \ne c_2$ we have that 
 $c_1-c_2 = ((\fnnf_1(P)-\fnnf_1(Q))^2,\ldots,(\fnnf_n(P)-\fnnf_n(Q))^2) \in \{0,1\}^n$, 
 for some $P,Q \in \{0,1\}^k,P \ne Q$. 
 The squaring operation is needed in order to correct those components which have become a $-1$ after the subtraction operation.
 Finally, the sum of all $(\fnnf_i(X)-\fnnf_i(\tilde{X}))^2$ is over the integers.
\end{proof}
%
\indent
We summarize in Algorithm \ref{algDistancePolyEval} the steps to obtain the distance distribution of a binary $(n,2^k)$-code $C$ given as a list of $2^k$ codewords 
(and thus also the minimum distance of $C$), by finding the evaluation vector of the distance polynomial $\dP_C$. 
We indicate with $C_{i,j}$ the $j$-th component of the $i$-th word of $C$, with $1 \le j \le n$ and $1 \le i \le 2^k$.
\begin{algorithm}[H]
\caption{To find the distance distribution $\underline{\dP}_C$ of a binary nonlinear code $C$.}
\label{algDistancePolyEval}
\begin{algorithmic}[1]
\REQUIRE{$c_1,\ldots,c_{2^k} \in C$}
\ENSURE{the evaluation vector $\underline{\dP}_C$ of $\dP_C$}
\STATE{$\fnnf_{j} \leftarrow \NNF$ of the binary vector $(C_{1,j},\ldots,C_{2^k,j})$ for $1 \le j \le n$}\label{algDPE:stepNNF}
\STATE{$\dP_C \leftarrow (\fnnf_1(X)-\fnnf_1(\tilde{X}))^2 + \ldots + (\fnnf_n(X)-\fnnf_n(\tilde{X}))^2$}\label{algDPE:stepWP}
\STATE{$\underline{\dP}_C \leftarrow$ Evaluation of $\dP_C$ over $\{0,1\}^{2k}$}\label{algDPE:stepEV}
\RETURN{$\underline{\dP}_C$} 
\end{algorithmic}
\end{algorithm}
%
%
\section{Complexity considerations}
  \label{sec:WeightComplexity}
First of all let us notice that given a binary $(n,2^k)$-code as a list of $2^k$ codewords, 
to find the weight distribution of a binary nonlinear code $C$ using brute force requires $n2^k$ bit operations, since we have to check each component of each codeword of $C$.
Similarly, to find the distance distribution, $n2^{2k}$ operations are needed. \\
We note that the operations involved in our following complexity estimates are over the integers, but 
the size of the integers involved in our operations is limited by $2^k$, and they have a sparse binary representation in the random case (they are sparse sums of powers of 2). \\
We now analyze the complexity of Steps~\ref{algWPE:stepNNF}, \ref{algWPE:stepWP}, and \ref{algWPE:stepEV} of Algorithm \ref{algWeightPolyEval} and \ref{algDistancePolyEval}. 
Then, due to the similarities of the two algorithms, we only concentrate on the first one.
We compare our method to compute the minimum weight of a binary code with brute force and, in the linear case, with the Brouwer-Zimmerman method (\cite{CGC-cd-book-zimmermann1996}).
We provide more emphasis on the comparison in the linear case, since no other methods than brute force are known in the nonlinear case, 
(with the exception of \cite{CGC-cd-art-pujol2012minimum,CGC-cd-art-villanueva2014efficient}).
\subsection{From list of codewords to defining polynomials in NNF}
\label{secCodeWordsToDefPol}
\begin{proposition}
 The overall worst-case complexity of determining the coefficients of the $n$ defining polynomials in NNF of the code $C$ given as a list of vectors is $O(nk2^k)$.
\end{proposition}
\begin{proof}
We want to find the NNF of the Boolean function whose truth table is given by a column of the binary matrix 
whose rows are the codewords of the code $C$. 
In \cite[Proposition 2]{CGC-cry-art-carlet1999} it is shown that to compute the NNF of a Boolean function in $k$ variables given its truth table requires $k2^{k-1}$ integer subtractions. Since we have to compute the NNF for $n$ columns the overall complexity is $O(nk2^k)$.
\end{proof}
%
\subsection{From defining polynomials to weight polynomial}
\label{secDefPolToWP}
\begin{proposition}
 The overall worst-case complexity of summing together all the defining polynomials in NNF is $O(n2^k)$.
\end{proposition}
\begin{proof}
 Each monomial $m$ in a defining polynomial is square-free, and since $m \in \ZZ[x_1,\ldots,x_k]$, then a defining polynomial can have no more than $2^k$ monomials. Since the defining polynomials are $n$, the proposition follows.
\end{proof}
\begin{remark}
 Clearly, the computational complexity of this steps decreases if the defining polynomials are sparse when considering their NNF.
\end{remark}
\subsection{From defining polynomials to distance polynomial}
\label{secDefPolToDP}
\begin{proposition}
 The overall worst-case complexity of Step \ref{algDPE:stepWP} of Algorithm \ref{algDistancePolyEval} is $O(n2^{2k})$.
\end{proposition}
\begin{proof}
The sum $\hat{f}_i = \fnnf_i(X)-\fnnf_i(\tilde{X})$ for $i=1,\ldots,n$ is just a concatenation of coefficients, 
where the coefficients of $\fnnf_i(\tilde{X})$ need to have their sign switched.\\
The polynomial obtained has $2^{k+1}$ terms in the worst case, and squaring it requires $2^{2(k+1)}$ integer multiplications and the same number of integer sums,
for a total of $2^{2k+3}$ integer operations. 
Since we have $n$ such polynomials $\hat{f}_i$, to compute their square requires $n2^{2k+3}$ integer operations. 
Each $\hat{f}_i$ has at most $2^{2k}$ terms, since $\hat{f}_i \in \ZZ[X,\tilde{X}] / \langle X^2-X,\tilde{X}^2-\tilde{X} \rangle$.
Summing all $\hat{f}_i$ together thus requires at most $n2^{2k}$ integer sums.
The overall worst-case complexity of Step \ref{algDPE:stepWP} of Algorithm \ref{algDistancePolyEval} is then
$$n2^{2k+3}+n2^{2k} = n2^{2k}(2^3+1)\,.$$
\end{proof}
\begin{remark}
Again, the complexity of this step is lower if the defining polynomials are sparse in their NNF.
If, for example, the nonzero coefficients of $\fnnf_i(X)$ are $\sim k$, so are the coefficients of $\fnnf_i(\tilde{X})$, 
and the squaring of $\hat{f}_i$ requires $\sim (2k)^2$ integer operations.
\end{remark}
\subsection{Evaluation of the weight and the distance polynomial}
\label{secWPtoEv}
Algorithm \ref{algFMTint} describes the fast M\"obius transform to compute the evaluation vector of a Boolean function $f$ in NNF in $k$ variables. \\
We use the following notation: the coefficient $c_{2^k}$ is the coefficient of the greatest monomial, 
i.e. of $x_1 \cdots x_k$, $c_{2^k-1}$ the coefficient of the second greatest monomial, and so on until $c_{1}$, which is the costant term. 
We provide Example \ref{exFMT} to clarify our notation.\\
Notice that the sum in Step \ref{stepUPDATE} is over our integers. 
If it was a sum in $\FF$ then we would obtain the truth table of $f$.
\begin{algorithm}[H]
\caption{Fast M\"obius transform for fast integer polynomial evaluation.}
\label{algFMTint}
\begin{algorithmic}[1]
\REQUIRE{vector of coefficients $c=(c_1,\ldots,c_{2^k})$}
\ENSURE{evaluation vector $e=(e_1,\ldots,e_{2^k})$}
\STATE{$e \leftarrow c$}
\FOR{$i=0,\ldots,k$}
  \STATE{$b \leftarrow 0$}
  \REPEAT
    \FOR{$x = b,\ldots,b+2^i-1$}
      \STATE{$e_{x+1+2^i} \leftarrow e_{x+1} + e_{x+1+2^i}$}
      \label{stepUPDATE}
    \ENDFOR
    \STATE{$b \leftarrow b + 2^{i+1}$}
  \UNTIL{$b = 2^k$}
\ENDFOR
\RETURN $e$
\end{algorithmic}
\end{algorithm}
\begin{example}\label{exFMT}
 Consider $k=3$ and lexicographical ordering with $x_1 \succ x_2 \succ x_3$. Let $f = 8x_1x_2x_3 + 3x_1 + 2$. Then 
 $c = (c_1,\ldots,c_{8}) = (2,0,0,0,3,0,0,8)$ and $e = (e_1,\ldots,e_{8}) = (2,2,2,2,5,5,5,13)$.
\end{example}
\begin{proposition}\label{thmEvWp}
 Evaluating the weight polynomial over the set $\{0,1\}^k$ has a computational cost of $O(k2^k)$.
\end{proposition}
\begin{proof}
  This is the cost of Algorithm \ref{algFMTint}, i.e. $k2^{k-1}$ integer sums.
\end{proof}
Similarly
\begin{proposition}\label{thmEvWp}
 Evaluating the distance polynomial over the set $\{0,1\}^{2k}$ has a computational cost of $O(k2^{2k})$.
\end{proposition}
%
\subsection{Comparison with brute-force method}
Because of the similarities of Algorithms \ref{algWeightPolyEval} and \ref{algDistancePolyEval}, 
we now concentrate our analysis only on Algorithm \ref{algWeightPolyEval}. 
All considerations we expose can be easily extended for Algorithm \ref{algDistancePolyEval}.\\
\begin{theorem}\label{thmWPcomplexity}
 Let $h$ be a positive integer. If the code $C$ is given as a set of \Bf's whose NNF have on average $2^k/h$ coefficients different from $0$, then computing the minimum weight of $C$ requires at most
 \begin{align*}
  \left( \frac{n}{h}+k \right) 2^k \,.
 \end{align*}
\end{theorem}
\begin{proof}
 By Proposition \ref{thmEvWp} computing the evaluation vector of the weight polynomial $\wP_C$ requires $k2^{k-1}$ integer sums using the fast M\"obius transform. To compute the weight polynomial we need to sum the $n$ defining polynomials $\fnnf_i,i=1,\ldots,n,$ in NNF. If each of these polynomials has on average $2^k/h$ coefficients then the complexity of computing $\wP_C$ requires $O(n \frac{2^k}{h})$ integer sums. So the final complexity is at most $(n/h)2^k+k2^{k-1}$.
\end{proof}
%
\begin{remark}
Our method is more efficient than brute force when $n/h + k < n$. This is very likely to happen for a random code of low information rate where $k \ll n$. 
If $k \sim n$ and the NNF is dense, then it is convenient to use brute force rather than our method.
\end{remark}
Notice also that if the sets of nonzero monomials of two polynomials in NNF are disjoint, then the sum of the two polynomials is simply their concatenation. So, if the defining polynomials of a code are ``disjoint'', then the cost of computing the weight polynomial is $O(1)$, and the final cost of finding the minimum weight becomes the cost of computing the evaluation of $\wP_C$, i.e. $O(k2^{k-1})$.\\
Fact \ref{thmComparisonVSBruteForceLinCase} shows that, for $n \gg k$, when the code is linear our method to compute the minimum nonzero weight (i.e. the distance of the code) given the set of the defining polynomials in NNF is more efficient than the classical method which uses brute force, given the list of the codewords of the code.
\begin{fact}[Comparison with brute force, linear case, $n \sim 2^k$] \label{thmComparisonVSBruteForceLinCase}
 Consider a random binary $[n,k]$-linear code $C$ such that $n \sim 2^k$. Then computing the weight distribution of $C$
 \begin{enumerate}
  \item given the list of its codewords and using brute force requires
  $O(2^{2k})$.
  \item given the list of the defining polynomials in NNF and finding the minimum of $\wP_C$ requires
  $O(2^{\frac{3}{2}k})$.
 \end{enumerate}
\end{fact}
\begin{proof}
  The complexity of finding the weight distribution of $C$ in case 1 is $O(n2^k) = O(2^{2k})$, since $n \sim 2^k$.\\
  The complexity of finding the weight distribution of $C$ in case 2 is $O((n/h+k)2^k)$ (by Theorem \ref{thmWPcomplexity}), where $n/h$ is the average number of nonzero coefficients of the NNF. If the linear code $C$ is random, then so are the random linear defining polynomials. A random linear function in $k$ variables has on average $k/2$ nonzero coefficient in ANF and thus $2^{k/2}-1$ nonzero coefficients in NNF 
  , i.e. $n/h \sim 2^{k/2}$, and
  $$O((n/h+k)2^k) = O((2^{k/2}+k)2^k) = O(2^{\frac{3}{2}k})\,.$$
\end{proof}
\begin{fact}[Comparison with brute force, nonlinear case, $n \sim 2^k$] \label{thmComparisonVSBruteForceNonLinCase}
Consider a random binary $(n,2^k)$-nonlinear code $C$ such that $n \sim 2^k$, and whose defining polynomials have on average $k/2$ nonzero coefficients in the ANF. 
Then computing the weight distribution of $C$ given the list of the defining polynomials in NNF and finding the minimum of $\wP_C$ requires
$O(2^{\frac{3}{2}k})$.
\end{fact}
\begin{proof}
 The arguments are the same as in the proof of Fact \ref{thmComparisonVSBruteForceLinCase}, except that this time the nonzero coefficients of the NNF are less than $2^{k/2}-1$. 
 This implies that in practice the overall complexity in this case is even lower, as shown in Table \ref{tabLinVSNonlinANF}.
\end{proof}
In Table \ref{tabLinVSNonlinANF} we show the coefficient of growth of the complexity of our method in three different cases. 
The first line shows the coefficient of growth of the brute force method applied to a linear code. The second line shows the coefficient of growth of our method applied to a linear code. In the third line our method is applied to a nonlinear code whose ANF representation is sparse, and in the last line nonlinear codes with dense ANF representation are considered.\\
For the comparison we choose for each $k$, $10$ random $(2^k,2^k)$-codes and $10$ random $(2^{k+1},2^{k+1})$-codes and compute the average times $t_1,t_2$ to compute the minimum weight in each case. 
Then we report the number $\log_2(t_1/t_2)$.\\
We can see, as expected, that our method performs best in the case of sparse nonlinear ANF.
\begin{table}[ht]
\begin{center}
\begin{tabular}{lcccc}
$k$                  & $8-9$ & $9-10$ & $10-11$ & $11-12$ \\
\hline
Brute-force Linear ANF & 1.93 & 1.98 & 2.00 & 1.99 \\ 
Linear ANF           &  1.32 &  1.38  &  1.53   &  1.61   \\
Sparse Nonlinear ANF &  0.89 &  1.12  &  1.32   &  1.38   \\
Dense  Nonlinear ANF &  2.09 &  2.03  &  2.04   &  2.08
\end{tabular}
\end{center}
\caption{Coefficients of growth of our method compared with brute force.}
\label{tabLinVSNonlinANF}
\end{table}
\subsection{Comparison with Brouwer-Zimmerman method for linear codes}
In the linear case the defining polynomials of a code $C$ clearly have a sparse ANF.
If a defining polynomial in $\FF[x_1,\ldots,x_k]$ is linear and with less than $k$ variables, than many coefficients of the NNF are $0$, precisely, the coefficients of the monomials containing the missing variable in the ANF. 
In this case the computation of the weight distribution of $C$ (and thus of the distance of $C$, since the code is linear) is faster than brute force.\\
In Table \ref{tabBrouZimm} we compare the time $t_1$ needed to compute the minimum weight $w$ of a linear code given as list of codewords with the MAGMA command
\begin{center}
\texttt{MinimumWeight(C:Method:=``Zimmerman'')},
\end{center}
with the time $t_2$ needed to compute $w$ when the code is given as a list of \Bf's in NNF using our method. The comparison has been done for 10 random linear codes fixing a pair $(k,n)$, with $n \gg k$. In the column $w_{av}$ the average minimum weight found is shown.\\
An AMD E2-1800 APU processor with $850$ MHz has been used for the computations.
\begin{table}[ht]
\begin{center}
\begin{tabular}{llllll}
$k$ & $n$ & $t_1$ & $t_2$ & $t_1/t_2$ & $w_{av}$\\
\hline
8  & $100k=800$  & $0.043$ & $0.007$ & $6.143$ & $360.1$ \\
8  & $150k=1200$ & $0.122$ & $0.012$ & $10.17$ & $554.1$ \\
8  & $200k=1600$ & $0.122$ & $0.015$ & $8.13$  & $745.2$ \\
8  & $250k=2000$ & $0.171$ & $0.011$ & $15.55$ & $935.0$ \\
\hline
9  & $100k=900$  & $0.833$ & $0.019$ & $4.368$ & $403.1$ \\
9  & $150k=1350$ & $0.116$ & $0.020$ & $5.800$ & $615.6$ \\
9  & $200k=1800$ & $0.277$ & $0.024$ & $11.54$ & $834.0$ \\
9  & $250k=2250$ & $0.256$ & $0.029$ & $8.828$ & $1050.0$ \\
\hline
10 & $100k=1000$ & $0.050$ & $0.031$ & $1.613$ & $448.3$ \\
10 & $150k=1500$ & $0.136$ & $0.041$ & $3.317$ & $687.5$ \\
10 & $200k=2000$ & $0.178$ & $0.050$ & $3.560$ & $922.7$ \\
10 & $250k=2500$ & $0.185$ & $0.056$ & $3.304$ & $1168.3$
\end{tabular}
\end{center}
\caption{Comparison with Brouwer-Zimmerman method.}
\label{tabBrouZimm}
\end{table}
We can see that there are cases, i.e. $(k,n) = (8,1200)$ or $(k,n) = (9,1800)$, where our method is 10 times faster than the Brouwer-Zimmerman method. This is not surprising, since the it is known that there are cases where brute force performs better than the Brouwer-Zimmerman method.\\
We also recall that Brouwer-Zimmerman method is probabilistic, while our method is deterministic.
\section{Binary codes whose cardinality is not a power of 2}
Algorithm \ref{algWeightPolyEval} can be modified to work also with binary codes whose cardinality is not a power of 2. 
We only mention two techniques that can be used.\\
A first method consist in expanding the code until it reaches a size of $2^k$.
The key observation is that the minimum weight vector of a list of vectors in $(\FF)^n$ (i.e. the codewords of $C$) is equal to the minimum weight vector of the same list concatenated to the list of some repeated words of $C$ (eventhough this new list is not a code anymore).\\
A second approach is to divide the code $C$ in subcodes whose cardinality is a power of 2. Then to each of these codes we can apply Algorithm \ref{algWeightPolyEval} and then take the minimum of all the results.
See \cite{CGC-cd-phdthesis-bellini} for details.
\bibliographystyle{amsalpha}
\bibliography{RefsCGC}

\end{document}